\documentclass[journal]{IEEEtran}
\usepackage{amsmath}
\usepackage{amssymb,amsthm,url}
\usepackage{cite}
\def\re#1#2{D(#1\vrt#2)}
\def\ext#1{2^{#1}}
\def\bbl#1{\big[#1\big]}
\def\mbl#1{\big\{#1\big\}}
\def\sbl#1{\big(#1\big)}
\def\nom#1{#1'}
\def\brk#1{|#1\ra\la#1|}
\def\lexp#1{\underline{D}(#1)}
\def\uexp#1{\overline{D}(#1)}
\newcommand{\vrt}{\Vert}
\newcommand{\hsst}{\sigma}
\newcommand{\cfn}{\mathcal{P}}
\newcommand{\mms}{\pi}
\newcommand{\stq}{\st_{\opq}}
\newcommand{\gfd}{F}
\newcommand{\pfd}{P}
\newcommand{\afa}{\searrow}
\newcommand{\afb}{\nearrow}
\newcommand{\lnt}{\ln2}
\newcommand{\uhst}{\mms_{\opg}}
\newcommand{\stsqg}{\hsst_{\ops\opq\opg}}
\newcommand{\smpr}{\ep^{*}}
\newcommand{\legp}{legitimate parties}
\newcommand{\reny}{R\'{e}nyi}
\newcommand{\ralp}{{\reny}-$\alpha$}
\newcommand{\len}{r}
\renewcommand{\Re}{\mathbb{R}}
\newcommand{\pr}{\mathrm{Pr}}
\newcommand{\tr}{\mathrm{Tr}}
\newcommand{\rank}{\mathrm{rank}}
\newcommand{\supp}{\mathrm{supp}}
\newcommand{\lgt}{\log_{2}}
\newcommand{\la}{\langle}
\newcommand{\ra}{\rangle}
\newcommand{\ot}{\otimes}
\newcommand{\iv}{J}

\newcommand{\ism}{U}
\newcommand{\sd}{d}
\newcommand{\sr}{D}
\newcommand{\td}{d_{1}}
\newcommand{\ord}{O}
\newcommand{\bll}{B}
\newcommand{\smin}{H_{\mathrm{min}}}
\newcommand{\dusd}{\sd(\opx|\opq)_{\st}}
\newcommand{\dusr}{\sr(\opx|\opq)_{\st}}
\newcommand{\dusrt}{\sr(\ops|\opq\opg)_{\stsqg}}

\newcommand{\hil}{{\cal{H}}}
\newcommand{\hlx}{\hil_{\opx}}
\newcommand{\hla}{\hil_{\opa}}
\newcommand{\hlb}{\hil_{\opb}}
\newcommand{\hle}{\hil_{\ope}}
\newcommand{\hlbp}{\hil_{\opbp}}
\newcommand{\hlc}{\hil_{\opc}}
\newcommand{\opq}{\opb}
\newcommand{\hlq}{\hil_{\opq}}
\newcommand{\hlg}{\hil_{\opg}}
\newcommand{\hls}{\hil_{\ops}}
\newcommand{\cqst}{\st_{\opx\opq}}
\newcommand{\rg}{{\cal{R}}}

\newcommand{\dd}{d}
\newcommand{\dma}{\dd_{\opa}}
\newcommand{\dmb}{\dd_{\opb}}
\newcommand{\ii}{i}
\newcommand{\nn}{n}
\newcommand{\ff}{f}
\newcommand{\rvg}{G}
\newcommand{\opg}{\rvg}
\newcommand{\g}{g}
\newcommand{\clg}{{\cal{G}}}
\newcommand{\rvx}{X}
\newcommand{\xx}{x}
\newcommand{\xxd}{\xx'}
\newcommand{\rvy}{Y}
\newcommand{\yy}{y}
\newcommand{\clx}{{\cal{X}}}
\newcommand{\rvs}{S}
\newcommand{\s}{s}
\newcommand{\cls}{{\cal{S}}}
\newcommand{\nqs}{\cls}
\newcommand{\sqs}{\cls_{\le}}
\newcommand{\pp}{p}
\newcommand{\prg}{\pp_{\g}}
\newcommand{\prx}{\pp_{\xx}}
\newcommand{\prxd}{\pp_{\xxd}}
\newcommand{\dl}{\delta}
\newcommand{\dlr}{\dl_{\len}}
\newcommand{\ep}{\epsilon}
\newcommand{\lep}{\underline{\ep}}
\newcommand{\uep}{\overline{\ep}}
\newcommand{\epab}{\lep}
\newcommand{\epb}{\uep}
\newcommand{\gm}{\gamma}
\newcommand{\kp}{\kappa}
\newcommand{\lm}{\lambda}
\newcommand{\rep}{\ren_{\ep}(\rvx|\rvy)}
\newcommand{\sg}{\sigma}
\newcommand{\sgb}{\sg_{\opb}}
\newcommand{\sguq}{\mms_{\opx}\ot\stq}
\newcommand{\st}{\rho}

\newcommand{\stx}{\st_{\xx}}
\newcommand{\sta}{\st_{\opa}}
\newcommand{\stab}{\st_{\opa\opb}}
\newcommand{\stbc}{\st_{\opb\opc}}
\newcommand{\stabc}{\st_{\opa\opb\opc}}
\newcommand{\stxy}{\st_{\opx\opy}}
\newcommand{\stxd}{\st_{\xxd}}
\newcommand{\sts}{\stb}
\newcommand{\stsc}{\stbc}
\newcommand{\stb}{\st_{\opb}}
\newcommand{\stp}{\hat{\st}_{\opb}}
\newcommand{\bstab}{\bar{\st}_{\opa\opb}}
\newcommand{\bstb}{\bar{\st}_{\opb}}
\newcommand{\cstb}{\check{\st}_{\opb}}
\newcommand{\hstb}{\hat{\st}_{\opb}}
\newcommand{\vn}{S}
\newcommand{\lvn}{\underline{\vn}}
\newcommand{\uvn}{\overline{\vn}}
\newcommand{\cvn}{\vn(\opa|\opb)_{\st}}
\newcommand{\ren}{R}
\newcommand{\rxro}{\Lambda}
\newcommand{\rxr}{\rxro_{\opb}}
\newcommand{\rxrc}{\rxro_{\opb\opc}}
\newcommand{\gqr}{\ren_{\ep}(\opa|\opb)_{\st}}
\newcommand{\ibr}{\bar}
\newcommand{\igr}{\ibr{\ren}_{\ep}}
\newcommand{\igqr}{\igr(\opa|\opb)_{\st}}
\newcommand{\igqrn}{\igr(\opa^{\nn}|\opb^{\nn})_{\st^{\ot\nn}}}
\newcommand{\lme}{\ren_{\ep}(\opx|\opb)_{\st}}
\newcommand{\lmr}{l}
\newcommand{\prj}{P}
\newcommand{\pbc}{\prj}
\newcommand{\opa}{A}
\newcommand{\opb}{B}
\newcommand{\opbp}{B'}
\newcommand{\opc}{C}
\newcommand{\ope}{E}
\newcommand{\ops}{S}
\newcommand{\opt}{T}
\newcommand{\opx}{X}
\newcommand{\opy}{Y}
\newcommand{\asg}{\st_{\s\g}}
\newcommand{\csg}{\opc_{\s\g}}
\newcommand{\xsg}{\opx_{\s\g}}
\newcommand{\css}{\opc_{+}}
\newcommand{\xss}{\opx_{+}}
\newcommand{\hlf}{1/2}
\newcommand{\id}{\mathbb{I}}
\newcommand{\idr}{\prj_{\st}}
\newcommand{\etz}{\eta_{0}}
\newcommand{\df}{\phi}
\newcommand{\rs}{\Delta}
\newcommand{\rsf}{\rs_{1}}
\newcommand{\rss}{\rs_{2}}
\newcommand{\upb}{\mu}
\newcommand{\lwb}{\lambda}
\newcommand{\ccn}{\check{c}}
\newcommand{\qch}{{\cal{F}}}
\newcommand{\snm}{\delta}
\newcommand{\nex}{e}
\newcommand{\gex}{\nex_{\gm}}
\newcommand{\eex}{\nex_{\ep}}
\newcommand{\elb}{\mu}
\newcommand{\eub}{\nu}
\newcommand{\opbs}{\opb^{*}}
\newcommand{\hlbs}{\hil_{\opbs}}
\newcommand{\stas}{\st^{*}}
\newcommand{\stbs}{\st_{\opbs}}
\newcommand{\dhy}{\hspace{0.3pt}{}^{\_}}
\newcommand{\ddt}{\hspace{-0.2pt}{}_{\cdot}\hspace{-0.4pt}}
\newcommand{\email}{yodai$\textcircled{\it a}$u$\dhy$aizu$\ddt$ac$\ddt$jp}
\theoremstyle{definition}
\newtheorem{theorem}{Theorem} 
\newtheorem{corollary}[theorem]{Corollary}
\newtheorem{proposition}[theorem]{Proposition}
\newtheorem{definition}[theorem]{Definition}
\begin{document}
\title{Randomness extraction via a quantum generalization
of the conditional collision entropy}
\author{Yodai~Watanabe,~\IEEEmembership{Member,~IEEE}%
\thanks{
This work was supported in part by 
JSPS Grants-in-Aid for Young Scientists (B) No. 21700021
and Scientific Research (C) Nos. 15K00020 and 19K11831.}
\thanks{The author is with 
Department of Computer Science and Engineering,
University of Aizu, Aizuwakamatsu, Fukushima 9658580, Japan
(e-mail: \email).}}%

\maketitle

\begin{abstract}
Randomness extraction against side information
is the art of distilling from a given source 
a key which is almost uniform conditioned on the side information.
This paper provides randomness extraction
against quantum side information
whose extractable key length is given by
a quantum generalization
of the collision entropy,
which is smoothed and conditioned differently from 
how this is done in existing schemes.
Based on the fact that the collision entropy
is not subadditive,
its optimization
with respect to additional side information is introduced,
and is shown to be asymptotically optimal.
The lower bound derived there for general states
is expressed as the difference between two unconditional entropies
and its evaluation reduces to an eigenvalue problem of two states,
which are the entire state and the marginal state of side information.
\end{abstract}

\begin{IEEEkeywords}
Randomness extraction;
Quantum collision entropy;
Extractable key length
\end{IEEEkeywords}

%
\IEEEpeerreviewmaketitle

\section{Introduction}
\label{intro}
Consider a classical system $\opx$ and a quantum system $\opq$
in a joint quantum state $\cqst$.
The task of randomness extraction
from the classical source $\opx$ against the quantum side information $\opq$
is to distill an almost random key $\rvs$ from $\opx$
by applying a classical channel $\cfn_{\opx\to\ops}$
from system $\opx$ to system $\ops$,
which results in a classical-quantum state 
$\hsst_{\ops\opq}=\cfn_{\opx\to\ops}(\cqst)$.
Here the randomness of $\rvs$
is measured by the trace distance
between the state $\hsst_{\ops\opq}$ of composite systems $\ops\opq$ 
and an ideal state $\mms_{\ops}\ot\stq$,
where $\mms_{\ops}$ is the maximally mixed state and
$\stq$ is the marginal state of $\cqst$.
A major application of randomness extraction is
privacy amplification \cite{bbcm95,rk05},
whose task is to transform
a partially secure key into a highly secure key
in the presence of an adversary with side information.

It has been shown that
a two-universal hash function \cite{cw79}
can be used to provide randomness extraction
against quantum side information,
in which the extractable key length
is lower-bounded by a quantum generalization
of the conditional min-entropy \cite{rr05}.
The extractable key length can also be given by
a quantum generalization of the conditional collision entropy
(conditional {\reny} entropy of order 2) \cite{rr05,rk05}.
It should be stated that
the collision entropy is lower-bounded by the min-entropy
and so gives a better extractable key length
than the min-entropy,
while the min-entropy has several useful properties
such as the monotonicity under quantum operations.

The way to consider
a tighter bound on the length
of an extractable almost random key
is to generalize entropies
by smoothing.
In fact,
the existing extractable key lengths
have been described by smooth entropies,
most of which are defined as the maximization
of entropies
with respect to quantum states 
within a small ball (see e.g. 
\cite{ha12,rr05,rk05,tcr09,tssr11}).
More precisely,
let $\hla$ and $\hlb$ be finite-dimensional Hilbert spaces,
and $\nqs(\hil)$ and $\sqs(\hil)$
denote the sets of normalized and sub-normalized quantum states
on a Hilbert space $\hil$, respectively;
then, for example,
the smooth min-entropy $\smin^{\ep}(\opa|\opb)_{\st}$
of system $\opa$ conditioned on system $\opb$
of a state $\st\in\sqs(\hla\ot\hlb)$
is defined by
\begin{align*}
\smin^{\ep}(\opa|\opb)_{\st}
&=\max_{\st'\in\bll^{\ep}(\st)}
\smin(\opa|\opb)_{\st'},\\
\smin(\opa|\opb)_{\st}
&=\max_{\sgb\in\nqs(\hlb)}\sup\{\lm|2^{-\lm}\id_{\opa}\ot\sgb\ge\st\},
\end{align*}
where $\id_{\opa}$ denotes the identity operator on $\hla$,
and $\bll^{\ep}(\st)
=\big\{\st'\in\sqs(\hil)\big|\pfd(\st,\st')\le\ep\big\}$
for $\ep>0$ and $\st\in\sqs(\hil)$
with $\pfd(\st,\st')=\sqrt{1-\gfd^{2}(\st,\st')}$
and $\gfd(\st,\st')=\tr\big|\sqrt{\st}\sqrt{\st'}\big|
+\sqrt{(1-\tr[\st])(1-\tr[\st'])}$.\footnote{%
$\gfd$ is called the generalized fidelity and 
$\pfd$ the purified distance~\cite{TCR10}.}
The hypothesis testing relative entropy~\cite{th13,wr12}
is based on the operator-smoothing~\cite{BD11,BD10},
which is different from the above standard smoothing
called the state-smoothing,
but is conditioned in the same way as above,
i.e. an operator of the form $\id_{\opa}\ot\sgb$
is introduced and the entropy is maximized
with respect to $\sgb$.

The contributions of this paper are summarized as follows: 
(i) This paper introduces a quantum generalization $\ren_{\ep}$
of the conditional collision entropy with smoothing parameter $\ep$
(Definition~\ref{grigr})
and shows that $\ren_{\ep}$ gives a lower bound on the key length of 
randomness extraction against quantum side information (Theorem~\ref{rethm}).
Since the smoothing and the conditioning for $\ren_{\ep}$
are not standard,\footnote{%
Here, smoothing is called standard
if a quantity of interest 
is maximized with respect to quantum states within a small ball,
and conditioning is called standard 
if an operator of the form $\id_{\opa}\ot\sg_{\opb}$ is introduced 
and a quantity of interest is maximized with respect to $\sg_{\opb}$.} 
the proof of the achievability of randomness extraction
in this paper 
might be rather independent of the existing ones.
(ii) This paper shows that the conditional collision entropy $\igr$ 
optimized with respect to additional side information\footnote{%
It follows from the monotonicity of the relative entropy 
that $\igr$ also gives a key length of 
randomness extraction against quantum side information.} 
(Definition~\ref{grigr}),
which automatically satisfies the strong subadditivity 
and so the data processing inequality (Proposition~\ref{dpineq}), 
is asymptotically optimal (Corollary~\ref{aopt}). 
This result demonstrates that 
the optimization of a conditional entropy 
with respect to additional side information 
can endow the entropy with not only the strong subadditivity 
but also the asymptotic optimality. 
(iii) This paper shows that $\igr$ has a general lower bound 
which is asymptotically optimal (Corollary~\ref{aopt})
and is expressed as
the difference between two unconditional entropies, 
each asymptotically approaching the von Neumann entropy, 
with an additive term of $\ord(\ep)$
for small $\ep$ (Theorem~\ref{gqrlwb}).
Since each unconditional entropy can be determined by
the eigenvalues of a given state,
the evaluation of the lower bound reduces to an eigenvalue problem
of two quantum states,
which are the entire state and the marginal state of side information.

\section{Preliminaries}
Let $\hil$ be a Hilbert space.
For an Hermitian operator $\opx$ on $\hil$
with spectral decomposition
$\opx=\sum_{\ii}\lm_{\ii}\ope_{\ii}$,
let $\{\opx\ge0\}$ denote the projection on $\hil$
given by
\begin{align}
\nonumber
\{\opx\ge0\}=\sum_{\ii:\lm_{\ii}\ge0}\ope_{\ii}.
\end{align}
The projections $\{\opx>0\}$, $\{\opx\le0\}$ and $\{\opx<0\}$
are defined analogously.
Let $\opa$ and $\opb$ be positive operators on $\hil$.
The trace distance $\td({\opa},{\opb})$
and the relative entropy $\re{\opa}{\opb}$
between $\opa$ and $\opb$
are defined as
\begin{align}
\nonumber
\td({\opa},{\opb})
&=\frac{1}{2}\tr[(\opa-\opb)(\{\opa-\opb>0\}
-\{\opa-\opb<0\})],\\
\nonumber
\re{\opa}{\opb}
&=\tr[\opa(\lgt\opa-\lgt\opb)],
\end{align}
respectively,
where $\lgt$ denotes the logarithm to base $2$.
The von Neumann entropy of $\opa$ is defined as
\begin{align}
\nonumber
\vn(\opa)=-\tr\opa\lgt\opa.
\end{align}
For an operator $\opx>0$,
let $\nom{\opx}$ denote the normalization of $\opx$;
that is, $\nom{\opx}=\opx/\tr[\opx]$.
It then follows from
$\vn(\nom{\opa})\le\lgt\rank\nom{\opa}=\lgt\rank\opa$
and
$\re{\nom{\opa}}{\nom{\opb}}\ge0$
that
\begin{align}
\label{vne_ineq}
\vn(\opa)
&\le\tr[\opa]\sbl{\lgt\rank\opa
-\lgt\tr[\opa]},
\\
\label{qre_ineq}
\re{\opa}{\opb}
&\ge
\tr[\opa]\sbl{\lgt\tr[\opa]-\lgt\tr[\opb]}.
\end{align}
More generally,
it can be shown that
inequality (\ref{vne_ineq}) holds for $\opa\ge0$
and 
inequality (\ref{qre_ineq}) holds for $\opa,\opb\ge0$
such that $\supp\opa\subset\supp\opb$,
by using the convention $0\lgt0=0$,
which can be justified by
taking the limit,
$\lim_{\ep\afa0}\ep\lgt\ep=0$.

Let $\ff$ be an operator convex function
on an interval $\iv\subset\Re$.
Let $\{\opx_{\ii}\}_{\ii}$ be a set of
operators on $\hil$ with their spectrum in $\iv$,
and
$\{\opc_{\ii}\}_{\ii}$ be a set of operators on $\hil$
such that $\sum_{\ii}\opc_{\ii}^{\dag}\opc_{\ii}=\id$,
where 
$\id$ is the identity operator on $\hil$.
Then Jensen's operator inequality
for $\ff$, $\{\opx_{\ii}\}_{\ii}$ and $\{\opc_{\ii}\}_{\ii}$
is given by
\begin{align}
\label{opjensen}
\ff\Big(\sum_{\ii}\opc_{\ii}^{\dag}\opx_{\ii}\opc_{\ii}\Big)
\le
\sum_{\ii}\opc_{\ii}^{\dag}\ff(\opx_{\ii})\opc_{\ii}
\end{align}
(see e.g. \cite{bh96,hp03}).

Let $\clx$ and $\cls$ be finite sets
and $\clg$ be a family of functions
from $\clx$ to $\cls$.
Let $\rvg$ be a random variable
uniformly distributed over $\clg$.
Then $\clg$ is called two-universal,
and
$\rvg$ is called a two-universal hash function \cite{cw79},
if
\begin{align}
\label{propuhf}
\pr[\rvg(\xx_{0})=\rvg(\xx_{1})]\le\frac{1}{|\cls|}
\end{align}
for every distinct $\xx_{0},\xx_{1}\in\clx$. 
For example,
the family of all functions from $\clx$ to $\cls$
is two-universal.
A more useful two-universal family is
that of all linear functions
from $\{0,1\}^{n}$ to $\{0,1\}^{m}$.
More efficient families,
which can be described using $O(n+m)$ bits
and have polynomial-time evaluating algorithms,
are discussed in \cite{cw79,wc81}.

Let $\hlx$ and $\hlb$ be Hilbert spaces,
and $\cqst$ be a classical-quantum state on $\hlx\ot\hlq$
given by
\begin{align}
\label{notcqst}
\cqst
=\sum_{\xx\in\clx}
\prx|\xx\ra\la\xx|\ot\stx,
\end{align}
where $\clx$ is a finite set
such that $\{|{\xx}\ra\}_{\xx\in\clx}$ forms an orthonormal basis of $\hlx$,
$\{\prx\}_{\xx\in\clx}$ is a probability distribution on $\clx$,
and $\stx\in\nqs(\hlb)$ 
for $\xx\in\clx$. 
Then the distance $\dusd$ from uniform
of system $\opx$ given system $\opq$ of a state $\st$
can be defined as
\begin{align}
\nonumber
\dusd
=\td(\st_{\opx\opq},\sguq),
\end{align}
where $\mms_{\opx}$ denotes the maximally mixed state on $\hlx$,
i.e. $\mms_{\opx}=\id_{\opx}/|\clx|$.

Instead of the trace distance,
another distance measure
may be used to define the distance from uniform.
For example,
the relative entropy
can be used
to define the distance from uniform of the form
\begin{align}
\label{dfure}
\dusr
=\re{\st_{\opx\opq}}{\sguq}.
\end{align}
Here, quantum Pinsker's inequality
$(2/\lnt)\big(\td(\st,\sg)\big)^{2}\le\re{\st}{\sg}$
(see~\cite{op93})
gives
\begin{align}
\nonumber
(2/\lnt)\big(\dusd\big)^{2}\le\dusr,
\end{align}
which ensures that
an upper bound on $\dusr$
also gives
an upper bound on $\dusd$.
Therefore,
in this paper, we will use $\dusr$,
instead of $\dusd$,
as the measure of the distance from uniform.

\section{Randomness extraction}
First, we introduce a quantum generalization of
the (smoothed) conditional collision entropy
and its optimization
with respect to additional side information.
\begin{definition}
\label{grigr}
Let $\hla$ and $\hlb$ be Hilbert spaces,
and $\stab$ be a quantum state on $\hla\ot\hlb$.
For $\ep\ge0$, 
the information spectrum collision entropy\footnote{%
This name of $\ren_{\ep}$ follows that of
the information spectrum relative entropy 
$D^{\ep}_{s}(\st\vrt\sg)=\sup\{R|\tr[\st\{\st\le2^{R}\sg\}]\le\ep\}$,
which can be considered as an entropic version 
of the quantum information spectrum, $\underline{D}$ and $\overline{D}$
(see~\cite{th13}).} %
$\gqr$ of system $\opa$ conditioned on system $\opb$
of a state $\st$
is given by
\begin{align}
\nonumber
\gqr=\sup_{\lm}
\big\{
\lm{\big|}
\tr\bbl{\mbl{\rxr
-\ext{-\lm}\sts^{2}\le0}\sts}\ge1-\ep
\big\},
\end{align}
where we have introduced
\begin{align}
\nonumber
\rxr
=\tr_{\opa}\bbl{\stab^{2}}.
\end{align}
Moreover, for $\ep\ge0$,
the information spectrum collision entropy $\igqr$
of system $\opa$ conditioned on system $\opb$ of a state $\st$
with optimal side information is given by
\begin{align*}
\igqr=\sup_{\hlc,\stabc:\tr_{\opc}[\stabc]=\stab}
\ren_{\ep}(\opa|\opb\opc)_{\st},
\end{align*}
where the supremum ranges over
all Hilbert spaces $\hlc$ and quantum states $\stabc$
on $\hla\ot\hlb\ot\hlc$
such that $\tr_{\opc}[\stabc]=\stab$.
\end{definition}
In contrast to the conditional von Neumann entropy 
$\cvn=\vn(\stab)-\vn(\stb)$,
$\gqr$ can increase when additional side information is provided;
that is,
\begin{align*}
\ren_{\ep}(\opa|\opb\opc)_{\st}>\gqr
\end{align*}
is possible 
(such side information for the classical collision entropy
is called spoiling knowledge~\cite{bbcm95}).
On the other hand,
$\igqr$ is optimized with respect to additional side information,
and so satisfies the following data processing inequality.
\begin{proposition}
\label{dpineq}
Let $\hla$, $\hlb$ and $\hlbp$ be Hilbert spaces,
and $\stab$ be a quantum state on $\hla\ot\hlb$.
Let $\qch$ be a trace preserving completely positive map
from system $\opb$ to system $\opbp$.
Then
\begin{align*}
\igqr\le\igr(\opa|\opbp)_{\qch(\st)}.
\end{align*}
\end{proposition}
\begin{proof}
The definitions of $\ren_{\ep}$ and $\igr$ at once give that
$\gqr$ is invariant under the adjoint action of an isometry on system $\opb$
and $\ren_{\ep}(\opa|\opb\opc)_{\st}\le\igqr$ for any $\stabc$.
Moreover,
the Stinespring dilation theorem (see~\cite{st55})
ensures that 
there exist a Hilbert space $\hle$ and
an isometry $\ism:\hlb\rightarrow\hlbp\ot\hle$ such that
\begin{align*}
\qch(\stb)=\tr_{\ope}\bbl{\ism\stb\ism^{\dag}}
\end{align*}
for any $\stb\in\nqs(\hlb)$.
Therefore,
for any Hilbert space $\hlc$ and
$\stabc\in\nqs(\hla\ot\hlb\ot\hlc)$
such that $\tr_{\opc}[\stabc]=\stab$,
\begin{align*}
\ren_{\ep}(\opa|\opb\opc)_{\st}
=\ren_{\ep}(\opa|\opbp\ope\opc)_{\ism\st\ism^{\dag}}
\le\igr(\opa|\opbp)_{\qch(\st)},
\end{align*}
and so
\begin{align*}
\igqr&=\sup_{\hlc,\stabc:\tr_{\opc}[\stabc]=\stab}
\ren_{\ep}(\opa|\opb\opc)_{\st}\\
&\le\igr(\opa|\opbp)_{\qch(\st)}.
\end{align*}
This completes the proof.
\end{proof}
We are now ready to state a main theorem.
Note that the monotonicity of the relative entropy
enables to replace $\gqr$ in this theorem by $\igqr$.
\begin{theorem}
\label{rethm}
Let $\clx$ and $\cls$ be finite sets,
and $\clg$ be a two-universal family of hash functions from $\clx$ to $\cls$.
Let $\hlx$, $\hls$ and $\hlg$ be Hilbert spaces 
of dimensions $|\clx|$, $|\cls|$ and $|\clg|$, respectively.
Let $\hlb$ be a Hilbert space, 
and $\cqst$ be a classical-quantum state on $\hlx\ot\hlb$.
Let $\uhst$ be the maximally mixed state on $\hlg$
independent of $\cqst$,
and suppose that the classical channel from $\opx$ to $\ops$ 
induced by the two-universal hash function $\opg$
maps $\cqst\ot\uhst$ to $\stsqg$.
Then
\begin{align}
\label{result}
\dusrt
\le
\ep\lgt\sbl{\dd|\cls|}
+\etz(\ep)
+\frac{\dl+\ep+\ep^{\hlf}}{\ln2}
\end{align}
for $\ep\ge0$,
where $\etz$ is a function on $[0,\infty)$
given by
\begin{align}
\label{defet}
\etz(\ep)=
\left\{
\begin{array}{cl}
-\ep\lgt\ep \quad& \mathrm{for}\ 0\le\ep\le1/2,\\
1/2 & \mathrm{for}\ \ep>1/2,
\end{array}
\right.
\end{align}
and we have introduced
\begin{align}
\nonumber
\dd=\rank\hspace{1pt}\sts
\quad
{\text{and}}
\quad
\dl=|\cls|\ext{-\lme}.
\end{align}
\end{theorem}
\begin{proof}
Direct calculation shows that
\begin{align}
\label{fsteq}
\begin{split}
&\dusrt\\
&=\sum_{\g,\s}
\prg\tr\bbl{\asg\lgt\asg}
-\tr\bbl{\sts\lgt\sts}
+\lgt|\cls|
\end{split}
\end{align}
with $\prg=1/|\clg|$,
where we have defined
\begin{align}
\label{defasg}
\asg=\sum_{\xx\in\g^{-1}(\s)}\prx\stx
\end{align}
(see (\ref{notcqst}) for the notation of a classical-quantum state).
Since the real function $-\lgt$
is operator convex on $(0,\infty)$
(see e.g. \cite{bh96}),
we can estimate
the first term of the right-hand side of (\ref{fsteq})
by applying Jensen's operator inequality as follows.
Let $\ff=-\lgt$
and introduce the operators $\xsg$ and $\csg$
by writing
\begin{align}
\nonumber
\xsg=\asg+\gm\id
\quad
{\text{and}}
\quad
\csg=(\prg\asg)^{\hlf}
\prj\stp^{-\hlf}
\end{align}
for $\gm>0$,
where we have defined
\begin{align}
\nonumber
\stp=\prj\sts\prj
\quad
{\text{and}}
\quad
\prj=\mbl{\rxr-\ext{-\len}\sts^{2}\le0}
\end{align}
for $\len<\lme$.
It readily follows that
$\xsg>0$
and
$\sum_{\s,\g}\csg^{\dag}\csg=\idr$,
where $\idr$ denotes the projection
onto the range of $\stp$.
Furthermore, let us define the operators
$\xss$ and $\css$ by
\begin{align}
\nonumber
\xss=\id
\quad
{\text{and}}
\quad
\css=\id-\idr,
\end{align}
so that
\begin{align}
\nonumber
\ff(\xss)=0
\quad
{\text{and}}
\quad
\sum_{\s,\g}\csg^{\dag}\csg+\css^{\dag}\css=\id.
\end{align}
Then
by Jensen's operator inequality (\ref{opjensen}),
\begin{align}
\nonumber
\ff\Big(\sum_{\s,\g}\csg^{\dag}\xsg\csg
+\css^{\dag}\xss\css\Big)
\le
\sum_{\s,\g}\csg^{\dag}\ff(\xsg)\csg.
\end{align}
Here, $\sum_{\s,\g}\csg^{\dag}\xsg\csg$
is an operator on the range $\rg(\stp)$ of $\stp$,
while $\css^{\dag}\xss\css$
is an operator on its orthogonal complement $\rg(\stp)^{\perp}$.
It thus follows that
\begin{align}
\nonumber
\stp^{\hlf}\ff\Big(
\sum_{\s,\g}\csg^{\dag}\xsg\csg
\Big)\stp^{\hlf}
\le
\sum_{\s,\g}\stp^{\hlf}\csg^{\dag}\ff(\xsg)\csg\stp^{\hlf},
\end{align}
which, in the limit $\gm\afa0$, leads to
\begin{align}
\label{qj_applied}
\begin{split}
&\sum_{\s,\g}
\prg\prj\asg^{\hlf}\big(\lgt\asg\big)\asg^{\hlf}\prj\\
&\le
\stp^{\hlf}
\bigg(
\lgt\sum_{\s,\g}
\prg
\stp^{-\hlf}
\prj\asg^{2}\prj
\stp^{-\hlf}
\bigg)
\stp^{\hlf}.
\end{split}
\end{align}

To estimate the right-hand side of (\ref{qj_applied}),
let us estimate
the sum $\sum_{\s,\g}\prg\prj\asg^{2}\prj$.
Substitution of (\ref{defasg}) into this sum
gives
\begin{align}
\nonumber
\sum_{\s,\g}
\prg\prj\asg^{2}\prj
=\sum_{\g,\xx,\xxd}
\prg\prx\prxd\pp(\g(\xx){=}\g(\xxd))
\prj\stx\stxd\prj.
\end{align}
Here,
we divide the sum of the right-hand side into two parts
so that one part consists of the terms with $\xx=\xxd$
and the other part consists of the remaining terms.
It follows from the definition of $\prj$ that
the former part can be bounded as
\begin{align}
\nonumber
\sum_{\g,\xx,\xxd:\xx=\xxd}
\prg\prx\prxd\pp(\g(\xx){=}\g(\xxd))
\prj\stx\stxd\prj
&=\prj\rxr\prj\\
\nonumber
&\le\ext{-\len}\prj\sts^{2}\prj.
\end{align}
By using (\ref{propuhf}), 
the latter part can also be bounded as
\begin{align}
\nonumber
\begin{split}
&\sum_{\g,\xx,\xxd:\xx\neq\xxd}
\prg\prx\prxd\pp(\g(\xx){=}\g(\xxd))
\prj\stx\stxd\prj\\
&=
\sum_{\xx,\xxd:\xx\neq\xxd}
\prx\prxd
\prj\stx\stxd\prj
\sum_{\g}
\prg\pp(\g(\xx){=}\g(\xxd))\\
&\le
\frac{1}{|\cls|}
\sum_{\xx,\xxd:\xx\neq\xxd}
\prx\prxd
\prj\stx\stxd\prj\\
&\le
\frac{1}{|\cls|}\prj\sts^{2}\prj.
\end{split}
\end{align}
The above two inequalities at once give
\begin{align}
\label{bndsum}
\sum_{\s,\g}
\prg\prj\asg^{2}\prj
\le
\frac{1}{|\cls|}
(1+\dlr)
\prj\sts^{2}\prj
\end{align}
with $\dlr=|\cls|\ext{-\len}$.
Note here that
$\lgt\asg\le0$ and $\asg\ge\asg^{\hlf}\prj\asg^{\hlf}$,
and so
\begin{align}
\nonumber
\sum_{\g,\s}
\prg\tr\bbl{\asg\lgt\asg}
\le
\sum_{\s,\g}
\prg
\tr\bbl{\asg^{\hlf}\prj\asg^{\hlf}\lgt\asg}.
\end{align}
Therefore,
by taking the trace of both sides of (\ref{qj_applied})
and then using (\ref{bndsum}) and $\tr[\stp]\le1$,
we obtain
\begin{align}
\nonumber
\dusrt
\le
(1-\tr[\stp])\lgt|\cls|
+\lgt(1+\dlr)+\rs,
\end{align}
where we have introduced
\begin{align}
\nonumber
\rs=
\tr\bbl{\stp\lgt\sbl{\stp^{-\hlf}
\prj\sts^{2}\prj\stp^{-\hlf}}}
-\tr[\sts\lgt\sts].
\end{align}
Furthermore,
by using $\tr[\stp]\ge1-\ep$ and
$\lgt(1+\xx)\le\xx/\ln2$ for $\xx\ge0$,
this inequality can be simplified to
\begin{align}
\label{semi-fin}
\dusrt
\le
\ep\lgt|\cls|
+\frac{\dlr}{\ln2}+\rs.
\end{align}

It remains to estimate $\rs$.
Let us write $\rs$
in the form $\rs=\rsf+\rss$,
where
\begin{align}
\nonumber
\rsf&=\tr\bbl{\stp\lgt\sbl{\stp^{-\hlf}
\prj\sts^{2}\prj\stp^{-\hlf}}}
-\tr[\stp\lgt\stp],\\
\nonumber
\rss&=\tr[\stp\lgt\stp]-\tr[\sts\lgt\sts].
\end{align}
First, we estimate the first part $\rsf$.
Let $\ops=\stp$ and 
$\opt=\stp^{-\hlf}\prj\sts^{2}\prj\stp^{-\hlf}$.
Since
\begin{align}
\nonumber
\opt-\ops
=\stp^{-\hlf}\prj\sts(\id-\prj)\sts\prj\stp^{-\hlf}
\ge0,
\end{align}
and hence $\supp\ops\subset\supp\opt$,
inequality (\ref{qre_ineq})
can be applied to yield
\begin{align}
\nonumber
\rsf=-\re{\ops}{\opt}
\le
\tr[\ops]\lgt\frac{\tr[\opt]}{\tr[\ops]}.
\end{align}
It is now convenient to define $\df=\tr[\opt-\ops]$,
which can be written as
\begin{align}
\nonumber
\df
&=\tr\bbl{\stp^{-\hlf}\prj\sts(\id-\prj)\sts\prj\stp^{-\hlf}}\\
\nonumber
&=\tr\bbl{(\id-\prj)\sts^{\hlf}\sts^{\hlf}\prj\stp^{-1}\prj\sts}.
\end{align}
Hence by
Schwarz's inequality,
\begin{align}
\nonumber
\df
\le
\big(
\tr\bbl{(\id-\prj)\sts(\id-\prj)}
\tr\bbl{\sts\prj\stp^{-1}\prj\sts}
\big)^{\hlf}.
\end{align}
By use of $\tr[(\id-\prj)\sts]\le\ep$ and
$\tr\bbl{\sts\prj\stp^{-1}\prj\sts}
=\tr[\opt]
=\tr[\ops]+\df$,
this inequality can be simplified to
$\df\le\ep^{\hlf}\sbl{\tr[\ops]+\df}^{\hlf}$,
which, together with $\tr[\ops]=\tr[\stp]\le1$,
gives
\begin{align}
\nonumber
\df\le
\frac{\ep+\sbl{\ep^{2}+4\ep\tr[\ops]}^{\hlf}}{2}
\le
\frac{\ep+\sbl{\ep+2\ep^{\hlf}}}{2}
=\ep+\ep^{\hlf}.
\end{align}
Therefore
\begin{align}
\label{estdo}
\rsf\le
\tr[\ops]\lgt\frac{\tr[\ops]+\df}
{\tr[\ops]}
\le\frac{\ep+\ep^{\hlf}}{\ln2}.
\end{align}

Next,
we estimate the second part $\rss$.
Let $\kp_{\prj}(\sts)=\prj\sts\prj+(\id-\prj)\sts(\id-\prj)$.
Since $\sts$ and $\kp_{\prj}(\sts)$ are density operators,
$\re{\sts}{\kp_{\prj}(\sts)}\ge0$,
and hence
\begin{align}
\nonumber
-\vn(\sts)
+\vn(\stp)
+\vn((\id-\prj)\sts(\id-\prj))
\ge0.
\end{align}
From this and (\ref{vne_ineq}),
\begin{align}
\label{estdt}
\begin{split}
\rss
&=-\vn(\stp)+\vn(\sts)
\le
\vn((\id-\prj)\sts(\id-\prj))\\
&\le
\ep\lgt\dd
+\etz(\ep),
\end{split}
\end{align}
where $\dd=\rank\sts$ and
$\etz$ is a monotone increasing function on $[0,\infty)$
defined by (\ref{defet}).\footnote{%
This inequality is a special case of Fannes inequality~\cite{F73}.}
Now,
the required inequality (\ref{result})
follows from
(\ref{semi-fin}), (\ref{estdo}) and (\ref{estdt})
with taking the limit $\len\afb\lme$.
This completes the proof.
\end{proof}

Let $\lmr^{\ep}_{D}(\opx|\opb)$ denote the maximal length
of randomness of distance $\ep$ from uniform
measured by the relative entropy $D$ (see (\ref{dfure})),
extractable from system $\opx$ given system $\opb$.
It then follows from this theorem that
\begin{align}
\label{grupb}
\lmr^{\ep'}_{D}(\opx|\opb)
\ge\igr(\opx|\opb)+\lgt\dl
\end{align}
with
\begin{align}
\label{eppdef}
\ep'=\ep\lgt\sbl{\dd|\clx|}
+\etz(\ep)
+\frac{\dl+\ep+\ep^{\hlf}}{\ln2}
\end{align}
(where we have used $|\cls|\le|\clx|$).
Since the smooth min-entropy $\smin^{\ep}$
upper-bounds the maximal length $\lmr^{\ep}$ of randomness 
(see e.g. \cite{th13}), 
this theorem also gives that
$\igr(\opx|\opb)$ is upper-bounded by $\smin^{\ep}$ as
\begin{align*}
\igr(\opx|\opb)&\le\lmr^{\ep'}_{D}(\opx|\opb)-\lgt\dl\\
&\le\smin^{\smpr}(\opx|\opb)_{\st}-\lgt\dl
\end{align*}
with $\smpr=(2(\lnt)\ep')^{1/4}$,
where we have used $\pfd(\st,\sg)\le\sqrt{2\td(\st,\sg)}$
(see~\cite{TCR10})
and quantum Pinsker's inequality.\footnote{%
The distance $\pfd(\ops|\opq\opg)_{\sg}$ from uniform in~\cite{th13} is defined 
by use of the purified distance
as $\pfd(\ops|\opq\opg)_{\sg}=\min_{\tau_{\opq}}
\pfd(\stsqg,\mms_{\ops}\ot\tau_{\opq}\ot\sg_{\opg})$,
which is upper-bounded by $\pfd(\stsqg,\mms_{\ops}\ot\sg_{\opq\opg})$.
}

It can be seen from inequality (\ref{result}) that 
the lower dimension $\dd$ of side information
gives the longer key length of randomness extraction.
Here, it may help to note that 
the monotonicity of the relative entropy allows us 
to assume that the dimension $\dd$ is controlled 
by {\legp} even in cryptographic applications 
where side information should be assumed 
under the full control of an adversary.
For the case where the {\legp} use high-dimensional states, 
one may consider a reduction to lower-dimensional states described below.
Let $\clx$ be a finite set.
For any set of quantum states $\{\stx\}_{\xx\in\clx}$ on $\hlb$,
one can construct a set of pure states 
$\{\stas_{\xx}\}_{\xx\in\clx}$ on $\hlbs$
such that there exists a trace preserving completely positive map
$\qch:\opbs\rightarrow\opb$
satisfying $\qch(\stas_{\xx})=\stx$ for all $\xx\in\clx$~\cite{cjw04}.
Therefore, for the case where $d=\rank\stb$ is high,
we may substitute $\ren(\opx|\opb)_{\qch(\st)}$
by $\ren(\opx|\opbs)_{\st}$
with $\st_{\opx\opbs}=\sum_{\xx}\prx\brk{\xx}\ot\stas_{\xx}$,
where $\ren(\opx|\opbs)_{\st}\le\ren(\opx|\opb)_{\qch(\st)}$ 
and $\rank\stbs\le|\clx|$.
(Here, the latter inequality is an advantage
but the former is a disadvantage of this reduction.)

In randomness extraction against
classical side information~$\rvy$,
the distance from uniform can be upper-bounded as
$\dusrt
\le\ep\lgt|\cls|+\dl/\ln2$
if $\rvg$ and $\rvy$ are independent~\cite{bbcm95},
and as
\[
\dusrt
\le\ep\lgt|\cls|+\frac{\dl+\ep}{\ln2}
\]
if $\rvy$ may depend on $\rvg$ 
\cite{yw07}.
(It should be stated that,
in quantum key distribution,
adversary's measurement can wait
until the choice of hash functions is announced,
and so adversary's information $\rvy$
may depend on the choice $\rvg$).
Here, we note that for a purely classical state $\stxy$,
$\ren(\opx|\opy)_{\st}$ becomes
\begin{align}
\nonumber
\rep=\sup_{\lm}\mbl{\lm|
\pr[\rvy\in\{\yy|\ren(\rvx|\rvy=\yy)\ge\lm\}]
\ge1-\ep},
\end{align}
where 
$\ren(\rvx|\rvy=\yy)
=-\lgt\sum_{\xx}
\pr[\rvx=\xx|\rvy=\yy]^{2}$.
Since
$\rep$ coincides with the (smoothed) conditional collision entropy
given by \cite{bbcm95,yw07},
$\gqr$ can be considered as its quantum generalization.
Hence it may be of interest to compare the results of these works.
It can be seen that
the upper bound given in this work (see (\ref{result}))
is larger than
that given in \cite{yw07} (see above) by
$\ep\lgt(\dd/\ep)+\ep^{\hlf}/\ln2$, 
which is $\ord(\ep^{\hlf})$ as $\ep\afa0$.

\section{Asymptotic optimality}
We first introduce two information spectrum entropies
which asymptotically approach the von Neumann entropy.
\begin{definition}
\label{defsie}
Let $\sta$ be a quantum state on a Hilbert space $\hla$.
Then, for $\ep\ge0$,
the information spectrum sup-entropy $\uvn_{\ep}(\opa)_{\st}$
and inf-entropy $\lvn_{\ep}(\opa)_{\st}$ of system $\opa$ of a state $\st$
are given by
\begin{alignat*}{2}
\uvn_{\ep}(\opa)_{\st}&=&
\inf_{\lm}&\big\{\lm{\big|}
\tr\bbl{\mbl{\st\ge2^{-\lm}}\st}\ge1-\ep\big\},\\
\lvn_{\ep}(\opa)_{\st}&=& \,
\sup_{\lm}&\big\{\lm{\big|}
\tr\bbl{\mbl{\st\le2^{-\lm}}\st}\ge1-\ep\big\},
\end{alignat*}
respectively.\footnote{%
This definition can be derived immediately
from the information spectrum relative entropy~\cite{th13}
via the formula $\vn(\st)=-\re{\st}{\id}$.}
\end{definition}
\begin{proposition}
\label{isvnasy}
Let $\sta$ be a quantum state on a Hilbert space $\hla$
of finite dimension $\dma$.
Then, for $\gm>0$,
\begin{align*}
\uvn_{\uep}(\opa^{\nn})_{\st^{\ot\nn}}
&\le\nn\big(\vn(\opa)_{\st}+\gm\big),\\
\lvn_{\lep}(\opa^{\nn})_{\st^{\ot\nn}}
&\ge\nn\big(\vn(\opa)_{\st}-\gm\big),
\end{align*}
with
\begin{align}
\label{defexp}
\uep=(1+\nn)^{\dma}2^{-\nn\uexp{\sta,\gm}}
\quad\text{and}\quad
\lep=(1+\nn)^{\dma}2^{-\nn\lexp{\sta,\gm}},
\end{align}
where we have introduced
\begin{align*}
\uexp{\st,\gm}
&=\inf_{\sg\in\nqs(\hil):\st\sg=\sg\st,
\vn(\sg)+\re{\sg}{\st}>\vn(\st)+\gm}
\re{\sg}{\st},\\
\lexp{\st,\gm}
&=\inf_{\sg\in\nqs(\hil):\st\sg=\sg\st,
\vn(\sg)+\re{\sg}{\st}<\vn(\st)-\gm}
\re{\sg}{\st},
\end{align*}
for $\st\in\nqs(\hil)$ and $\gm>0$.
\end{proposition}
\begin{proof}
Note that $\uvn_{\ep}(\opa)_{\st}$ and $\lvn_{\ep}(\opa)_{\st}$ 
can be described by the probability distribution
induced by the eigenvalues of $\sta$.
Hence, the proposition is a direct consequence of 
Sanov's theorem (see e.g.~\cite{ct06}),
which gives that, in our notation,
\begin{align*}
\tr\bbl{\st^{\ot\nn}\{2^{-\nn\elb}
<\st^{\ot\nn}<2^{-\nn\eub}\}}
\le(1+\nn)^{\dd}2^{-\nn D(\st;\elb,\eub)}
\end{align*}
with
\begin{align*}
D(\st;\elb,\eub)
=\inf_{\sg\in\nqs(\hil):\st\sg=\sg\st,
\eub<\vn(\sg)+\re{\sg}{\st}<\elb}
\re{\sg}{\st}
\end{align*}
for $\st\in\nqs(\hil)$,
where $\hil$ is a Hilbert space of finite dimension~$\dd$.
\end{proof}

Next, we give a general lower bound on $\igr$
in terms of the two information spectrum entropies introduced above,
and then show the asymptotic optimality of $\igr$.
Since each information spectrum entropy is determined
by the eigenvalues of a quantum state, the evaluation of the lower bound
reduces to the eigenvalue problem of
two quantum states $\stab$ and $\sta$.
\begin{theorem}
\label{gqrlwb}
Let $\hla$ and $\hlb$ be Hilbert spaces,
and $\stab$ be a quantum state on $\hla\ot\hlb$.
Then, for $\epab,\epb>0$,
\begin{align*}
\igqr\ge\lvn_{\epab}(\opa\opb)_{\st}
-\uvn_{\epb}(\opb)_{\st}
+\lgt\big(1-\epab^{1/2}\big)
\end{align*}
with
\begin{align*}
\nonumber
\ep=\epab^{1/2}+\epab+\epb.
\end{align*}
\end{theorem}
\begin{proof}
Let $\hlc$ be a two-dimensional Hilbert space
with an orthonormal basis $\{|0\ra,|1\ra\}$.
For $\epab>0$,
define a quantum state $\stabc$ 
on $\hla\ot\hlb\ot\hlc$ by
\begin{align*}
\stabc=\bstab\ot\brk{1}+(\stab-\bstab)\ot\brk{0},
\end{align*}
where we have introduced
\begin{align*}
\bstab=\stab\big\{\stab\le\upb\big\}
\end{align*}
with $\upb=2^{-\lvn_{\epab}(\opa\opb)_{\st}}$.
It is clear from this definition that
\begin{align*}
\tr_{\opc}[\stabc]=\stab
\quad\text{and}\quad
\bstab^{2}\le\upb\bstab.
\end{align*}
Also,
it follows from the definition of $\lvn_{\ep}$ that
\begin{align}
\label{trrb}
\tr[\bstab]\ge1-\epab.
\end{align}
Moreover,
for $\epb>0$,
define $\hstb$ and $\cstb$ by
\begin{align*}
\hstb&=\stb\big\{\stb\ge\lwb\big\},\\
\cstb&=\big\{\stb\ge\lwb\big\}\bstb\big\{\stb\ge\lwb\big\},
\end{align*}
with $\lwb=2^{-\uvn_{\epb}(\opb)_{\st}}$ and $\bstb=\tr_{\opa}[\bstab]$.
It then follows from $\bstb\le\stb$ and (\ref{trrb}) that
\begin{align*}
\cstb\le\hstb
\quad\text{and}\quad
\tr\big[\hstb-\cstb\big]\le\epab,
\end{align*}
and so for $\ccn=1-\epab^{1/2}$,
\begin{align*}
\tr\big[\hstb\big\{\cstb<\ccn\lwb\big\}\big]
&\le\tr\big[\hstb\big\{\cstb<\ccn\hstb\big\}\big]\\
&=\tr\big[\hstb\big\{\epab^{1/2}\hstb<\hstb-\cstb\big\}\big]\\
&<\tr\big[\epab^{-1/2}\big(\hstb-\cstb\big)\big]
\le\epab^{1/2}.
\end{align*}
Hence,
\begin{align*}
\tr\big[\hstb\big\{\cstb\ge\ccn\lwb\big\}\big]
>\tr\big[\hstb\big]-\epab^{1/2}.
\end{align*}
Here,
on noting that 
$\cstb$ commutes with
$\big\{\stb\ge\lwb\big\}$,
let us introduce the projection $\pbc$
on $\hlb\ot\hlc$ defined by
\begin{align*}
\pbc=
\big\{\stb\ge\lwb\big\}\big\{\cstb\ge\ccn\lwb\big\}
\ot\brk{1}.
\end{align*}
It can be seen from this definition that
\begin{align*}
\tr\bbl{\stsc\pbc}
&>\tr\big[\hstb\big]-\epab^{1/2}-\tr\big[\stb-\bstb\big]\\
&\ge1-\epb-\epab^{1/2}-\epab.
\end{align*}
Moreover, since
\begin{align*}
\pbc\rxrc\pbc
=\pbc\tr_{\opa}\big[\bstab^{2}\ot\brk{1}\big]\pbc
\le\upb\pbc\stsc\pbc
\end{align*}
and
\begin{align*}
\pbc\stsc^{2}\pbc\ge\pbc\stsc\pbc\stsc\pbc
=\pbc(\cstb^{2}\ot\brk{1})\pbc
\ge\ccn\lwb\pbc\stsc\pbc,
\end{align*}
it follows that
\begin{align*}
\pbc(\rxrc-2^{-\len}\stsc^{2})\pbc\le0
\end{align*}
for
$\len\le\lvn_{\epab}(\opa\opb)_{\st}
-\uvn_{\epb}(\opb)_{\st}+\lgt\ccn$.
Hence, $\pbc\le\{\rxrc-2^{-\len}\stsc\le0\}$
and so
\begin{align*}
\nonumber
\igqr
&\ge\ren_{\ep}(\opa|\opb\opc)_{\st}\\
&\ge\lvn_{\epab}(\opa\opb)_{\st}
-\uvn_{\epb}(\opb)_{\st}
+\lgt\big(1-\epab^{1/2}\big)
\end{align*}
for
\begin{align}
\nonumber
\ep=\epab^{1/2}+\epab+\epb.
\end{align}
This completes the proof.
\end{proof}
\begin{corollary}
\label{aopt}
Let $\hla$ and $\hlb$ be Hilbert spaces
of finite dimensions $\dma$ and $\dmb$, respectively,
and $\stab$ be a quantum state on $\hla\ot\hlb$.
Then,
\begin{align*}
\lim_{\nn\rightarrow\infty}
\frac{1}{\nn}
\igqr\ge\cvn
\end{align*}
for $\ep$ converging to $0$ as $\nn\rightarrow\infty$.
\end{corollary}
\begin{proof}
It follows from Theorem~\ref{gqrlwb} and Proposition~\ref{isvnasy} that
\begin{align*}
\igqrn\ge\nn\big(\cvn-2\gm\big)
+\lgt\big(1-\epab^{1/2}\big)
\end{align*}
with
\begin{align*}
\nonumber
\ep=\epab^{1/2}+\epab+\epb,
\end{align*}
where $\epab$ and $\epb$ are given by~(\ref{defexp}).
Now, suppose that
$\st$, $\sg$ and $\gm$ satisfy the condition 
in the definition of $\overline{D}$ or $\underline{D}$.
Then
\begin{align*}
\gm&<|\vn(\st)-\vn(\sg)|\pm\re{\sg}{\st}\\
&\le\td({\st},{\sg})\dim\hil
-\td({\st},{\sg})\lgt\td({\st},{\sg})
\pm\re{\sg}{\st}\\
&\le{\re{\sg}{\st}}^{\frac{1-\snm}{2}}
\end{align*}
with $\snm>0$,
for sufficiently small $\re{\sg}{\st}$,
where the second inequality follows from 
Fannes inequality~\cite{F73}
and the third one from
quantum Pinsker's inequality and 
$\lim_{x\afa0}x^{\snm}\lgt x=0$ for $\snm>0$.
Therefore,
we can take
\begin{align*}
\gm=\nn^{-\gex},\
\epab=(1+\nn)^{\dma\dmb}2^{-\nn^{\eex}},\
\epb=(1+\nn)^{\dmb}2^{-\nn^{\eex}}
\end{align*}
for $\gex$ and $\eex$ such that
\begin{align*}
\gex,\eex>0
\quad\text{and}\quad
2\gex+\eex<1,
\end{align*}
from which the corollary follows.
\end{proof}

\section{Concluding remarks}
There have been many works on 
the quantities characterizing randomness extraction and 
it has been shown that these quantities have several useful properties 
and they are equivalent up to additive terms of $\ord(\lgt\ep)$
(see e.g. \cite{rr05,th13}). 
Hence, $\igr$ should also be examined in more detail 
to clarify its properties and relations to other quantities. 
Also, 
it may be of interest to consider further extensions
and generalizations of the results~(i)--(iii)
(see the last paragraph of Section~\ref{intro}).
For example, 
$\igr$ is defined for fully quantum states, but 
this work gives its operational meaning 
only for classical-quantum states; 
hence it remains to examine its operational meaning 
for fully quantum states.
Moreover,
since the collision entropy is a special case of 
the {\ralp} entropies, 
it is of interest to consider analogous generalizations
of $\igr$ and investigate their operational meanings. 
Furthermore, it may be natural to consider
the possibility to extend the result~(ii) 
to other conditional entropies.
Regarding the result~(iii),
it remains to derive an upper bound on $\igr$ 
consistent with the lower bound
in Theorem~\ref{gqrlwb}. 
Moreover,
since an asymptotic expansion 
of the maximal length $\lmr^{\ep}$ of randomness
with an optimal second-order term 
(for fixed distance $\ep$ from uniform)
has been derived~\cite{th13},
it may be of interest to examine
how close $\frac{1}{\nn}\igr$ is to this optimum,
in particular when the classical large deviation theory 
giving an optimal second-order asymptotics~\cite{H08}
is applied instead of Sanov's theorem.
Finally, since 
in many applications such as those in cryptography,
$\ep$ should converge to $0$ faster than any polynomial $\nn^{-c}$
for sufficiently large $\nn>\nn_{c}$, 
it is also of interest to examine the possibility
of further improvement in the second-order asymptotics
for $\ep$ converging sufficiently fast.%
\footnote{%
We note that the additive term $\lgt(1-\epab^{1/2})$ 
in the lower bound is not $\ord(\lgt\ep)$ but $\ord(\ep)$.}

\section*{Acknowledgement}
The author is grateful to reviewers for their helpful comments.

\ifCLASSOPTIONcaptionsoff
  \newpage
\fi






\end{document}